\documentclass[11pt]{article}

\usepackage{graphicx}
\usepackage{dcolumn}
\usepackage{bm}

\usepackage{mathrsfs}  

\usepackage{appendix}

\usepackage[utf8]{inputenc}
\usepackage[T1]{fontenc}
\usepackage{color}
\usepackage{physics}
\usepackage{fullpage}
\usepackage{amsthm}
\usepackage{amssymb}
\usepackage{mathtools}
\usepackage{authblk}
\usepackage{thm-restate}
\usepackage{bbm}
\usepackage[colorlinks,linkcolor=blue,citecolor=blue,urlcolor=magenta,filecolor=cyan]{hyperref}
\usepackage[capitalize]{cleveref}

\newtheorem{theorem}{Theorem}
\newtheorem{lemma}[theorem]{Lemma}

\def\CO{\mathcal{O}}
\def\polylog{\mathrm{polylog}}
\def\poly{\mathrm{poly}}

\def\bbc{\mathbb{C}}

\newcommand{\comment}[1]{}

\title{Simulating Markovian open quantum systems \\ using higher order series expansion}

\author[*]{Xiantao Li}

\author[$\dag$]{Chunhao Wang}
\affil[*]{Department of Mathematics, Pennsylvania State University}
\affil[$\dag$]{Department of Computer Science and Engineering, Pennsylvania State University}
\affil[ ]{Email: xiantao.li@psu.edu; cwang@psu.edu}

\date{}
\begin{document}
\maketitle

\begin{abstract}
  We present an efficient quantum algorithm for simulating the dynamics of Markovian open quantum systems. The performance of our algorithm is similar to the previous state-of-the-art quantum algorithm, i.e., it scales linearly in evolution time and poly-logarithmically in inverse precision. However, our algorithm is conceptually cleaner, and it only uses simple quantum primitives without compressed encoding. Our approach is based on a novel mathematical treatment of the evolution map, which involves a higher-order series expansion based on Duhamel's principle and approximating multiple integrals using scaled Gaussian quadrature. Our method easily generalizes to simulating quantum dynamics with time-dependent Lindbladians. Furthermore, our method of approximating multiple integrals using scaled Gaussian quadrature could potentially be used to produce a more efficient approximation of time-ordered integrals, and therefore can simplify existing quantum algorithms for simulating time-dependent Hamiltonians based on a truncated Dyson series.
\end{abstract}

\section{Introduction}
The last few decades have witnessed the exciting progress in quantum information science to understand and utilize systems that exhibit quantum properties. In the meantime, quantum algorithms for simulating quantum dynamics have received extensive attention. This is because such simulation algorithms are critical tools in many physics applications, and they have the potential to become the first application (if it is not factoring integers!) once large-scale fault-tolerant quantum computers are available. In fact, simulating quantum dynamics was one of the original motivations Feynman proposed quantum computers~\cite{Feynman1982}, who realized the unfavorable scaling for classical algorithms for this task and foresaw the power of quantum computation back in 1982.

Up to now, the majority of the research for simulating quantum dynamics is focused on the ``Hamiltonian regime'', where the system is governed by  Schr\"odinger evolution and it has no interaction with the environment. Such idealized systems are often referred to as \emph{closed systems}. However, if one believes that the universe is a closed system, then it is reasonable to assume that every quantum system, as a subsystem of the whole universe, is an \emph{open system} because every realistic system is coupled to an uncontrollable environment in a non-negligible way. For example, we always model quantum gates as unitary matrices, while their implementations are always subject to noise induced by the environment no matter how hard one pushes the technology forward.

A key challenge in simulating the dynamics of open quantum systems is the lack of a microscopic description of the dynamics influenced by the physical law of the environment. Even if such a description exists, the degrees of freedom will involve numerous information, which would exceed the capacity of quantum computers.
Fortunately, for a special class of open quantum systems, their dynamics can be \emph{fully} described by operators acting on the system. This special class captures the scenario when the system is weakly coupled to the environment and the dynamics of the environment occur at a much faster rate than the system. Intuitively, the environment is fast enough so that the information only flows from the system to the environment and there is no information flowing back. Precisely due to such Markovianity, these open systems are called \emph{Markovian open quantum systems}. Specifically, such dynamics are described in terms of the density matrix $\rho$ by the~differential~equation
\begin{align}
  \label{eq: lindb}
  \frac{\dd}{\dd t}\rho = \mathcal{L}(\rho) \coloneqq -i[H, \rho] + \sum_{j=1}^m\left(L_j\rho L_j^{\dag} - \frac{1}{2}\left\{L_j^{\dag}L_j, \rho\right\}\right)
\end{align}
which is referred to as the \emph{Lindblad equation}~\cite{lindblad1976generators,gorini1976completely}. The superoperator $\mathcal{L}$ is called the \emph{Lindbladian}, and the $L_j$'s are often called the \emph{jump operators}. The solution to the Lindblad equation is given by
\begin{align}
  \rho_t = e^{\mathcal{L}t}(\rho_0).
\end{align}
Here, the superoperator $e^{\mathcal{L}t}$ is a quantum channel for all $t \geq 0$.

It turns out that Markovian open quantum systems are general enough to model many realistic quantum systems in quantum physics~\cite{LCDFGZ1987,Weiss2012}, quantum chemistry~\cite{MK2008,Nitzan2006}, and quantum biology~\cite{DGV2012,HP2013,MREKTA2012}. Computationally, such systems also arise in the context of entanglement preparation~\cite{KBDKMZ2008,KRS2011,RRS_2016}, thermal state preparation~\cite{KB2016}, quantum state engineering~\cite{VWC2009}, and modeling the noise of quantum circuits~\cite{MPGC2013,OG19,SYT21}.

The first quantum algorithm for simulating Markovian open quantum systems was presented by Kliesch et al.~in~\cite{KBG11} in 2011 and the complexity has scaling $O(t^2/\epsilon)$ for evolution time $t$ and precision $\epsilon$. In 2017, Childs and Li~\cite{CL17} constructed an improved algorithm with cost  $O(t^{1.5}/\sqrt{\epsilon})$. Cleve and Wang~\cite{CW17} pushed the study further by reducing the complexity to nearly optimal: $O(t\,\polylog(t/\epsilon))$, which was the first to achieve the complexity that scales linearly in $t$ and poly-logarithmically in $1/\epsilon$ --- that is exponentially better than previous approaches. Recently, researchers have explored these simulation algorithms in various scenarios such as simulating heavy-ion collisions~\cite{dJMM+21}, simulating the non-equilibrium dynamics in the Hubbard model~\cite{TGH22}, simulating the non-equilibrium dynamics in the Schwinger model~\cite{dJLM2+21}, and preparing thermal states~\cite{RWW22}.

The quantum algorithm by Cleve and Wang~\cite{CW17} is based on the first-order approximation of $e^{\mathcal{L}t}$, which can be further approximated by a completely positive map whose Kraus operators involve $H$ and $L_j$'s. Due to the inefficiency of the first-order approximation, the building blocks (the implementation of linear combination unitaries) of~\cite{CW17} need to be repeated many times to simulate a constant-time evolution, which tends to break the poly-logarithmic dependence on $1/\epsilon$.  However, it was shown in~\cite{CW17} that the state of the control qubits of the building blocks concentrates to low-Hamming weight states. Thus a compression scheme had to be employed in~\cite{CW17} to exponentially reduce the uses of the building blocks.

In the literature of Hamiltonian simulation, there is an elegant quantum algorithm that uses a truncated Taylor series~\cite{BCC15}. This algorithm is conceptually much simpler than  its first-order approximation predecessor~\cite{BCC17} while keeping the same efficiency. Thus a natural question arises: Can we generalize the truncated Taylor series approach to simulating Lindblad evolution to obtain a much simpler algorithm? It was not known how to achieve this due to the obstacle that higher powers of the Lindbladian are too intricate to keep track of its completely positiveness, which is the key to implementing a superoperator. To demonstrate this challenge, consider the expression of $\mathcal{L}^2$. For simplicity, let us assume $H = 0$ and $m=1$. We have
\begin{align}
  \begin{aligned}
    \mathcal{L}^2(\rho) &= L^2\rho L^{\dag 2} - \frac{1}{2}LL^{\dag}L\rho L^{\dag} - \frac{1}{2}L\rho L^{\dag}LL^{\dag} - \frac{1}{2}L^{\dag}L^2\rho L^{\dag} + \frac{1}{4}L^{\dag}LL^{\dag}L\rho  \\
                        &-\frac{1}{2}L\rho L^{\dag 2}L + \frac{1}{2}L^{\dag}L\rho L^{\dag}L + \frac{1}{4}\rho L^{\dag}LL^{\dag}L.
  \end{aligned}
\end{align}
As the above shows, it is highly nontrivial to maintain the completely positive structure in the Taylor series 
  $e^{\mathcal{L}} \approx \sum_{\ell=0}^K \frac{\mathcal{L}^{\ell}}{\ell!}$
even for this overly simplified case where $H=0$ and $m=1$.

In this paper, we present a quantum algorithm that takes advantage of a higher-order series expansion based on Duhamel's principle (this principle is briefly discussed in \cref{sec:Duhamel} for readers not familiar with this subject). Our quantum algorithm is conceptually simple and it only contains straightforward applications of simple quantum primitives such as oblivious amplitude amplification for isometries and linear combinations of unitaries (LCU) for channels (presented in the language of block-encodings)~\cite{CL17}. Our approach is inspired by a classical algorithm by Cao and Lu~\cite{CL21} based on the Duhamel's principle. The basic idea is to separate the Lindblad generator into two groups, the first group of which can be expressed as a matrix exponential that immediately induces a completely positive map. By applying the Duhamel's principle repeatedly, a series expansion with arbitrary order of accuracy can be obtained. We prove a rigorous error bound for this truncation. This procedure exhibits some level of resemblance to the time-dependent Hamiltonian simulation method using Dyson series~\cite{kieferova2019simulating}. However, an important focus in the simulation of Lindblad dynamics, due to the presence of jump operators, is to maintain the completely positive property. In addition, we apply Gaussian quadrature, which for any fixed number of quadrature points, has the optimal order of accuracy, to treat the multiple integrals in the series expansion. This approach, compared with the rectangle rule used in~\cite{kieferova2019simulating} and the mid-point rule used in~\cite{CL21},  compressed drastically the number of terms in the Kraus form of the completely positive maps. The other added advantage is that it completely eliminates the need for time clocking, which requires either a compression scheme or a quantum sorting network to implement. 

We consider a very general model of input, namely, the operators are given by their block-encodings. Informally, a unitary $U_A$ is the block-encoding of $A$ with normalizing factor $\alpha$ if the top-left block of $U_A$ is $A/\alpha$. This input model abstract but general enough to assume for most realistic physical models. In fact, traditional input models such as local Hamiltonians, sparse Hamiltonians, and a linear combination of tensor product of Paulis can all be efficiently converted to block-encodings. Suppose the operators $H$, and $L_j$'s are given as block-encodings with corresponding normalizing factors $\alpha_0$, $\alpha_j$'s, respectively. We define the following norm for the purpose of normalizing the magnitude of the Lindbladian in \cref{eq: lindb}.
\begin{align}
  \norm{\mathcal{L}}_{\mathrm{be}} \coloneqq \alpha_0 + \frac{1}{2}\sum_{j=1}^m\alpha_j^2.
\end{align}
Note that a similar norm $\norm{\mathcal{L}}_{\mathrm{pauli}}$ was defined in~\cite{CL17}, which is a special case of $\norm{\mathcal{L}}_{\mathrm{be}}$ in the context of linear combination of unitaries input model.
Our main result is stated as follows.
\begin{theorem}[Informal version of \cref{thm:main}]
  For a Lindbladian $\mathcal{L}$ with $m$ jump operators. Suppose we are given a block-encoding $U_H$ of $H$, and a block-encoding $U_{L_j}$ for each $L_j$. For all $t,\epsilon > 0$, there exists a quantum algorithm that approximates $e^{\mathcal{L}t}$ in terms of the diamond norm using
    $O(\tau\,\mathrm{polylog}(t/\epsilon))$
  queries to $U_H$ and $U_{L_j}$ and 
    $O(m\tau\,\mathrm{polylog}(t/\epsilon))$
    additional 1- and 2-qubit gates, where $\tau \coloneqq t\norm{\mathcal{L}}_{\mathrm{be}}$.
\end{theorem}

Our approach trades off mathematical simplicity for technical conciseness. In fact, the majority of the analysis is devoted to proving the bound of the truncated series, and the accuracy for using a scaled Gaussian quadrature to approximate each layer of a multiple integral. Once the mathematical treatment is established, we obtain an approximate map for $e^{\mathcal{L}t}$ that is completely positive,  represented in terms of Kraus operators. Then, it is straightforward to use simple quantum primitives including LCU for channels and oblivious amplitude amplification for isometries to implement this completely positive map. Moreover, it is more direct to obtain the gate complexity that scales linearly in $m$, for which the dependence was $O(m^2)$ as presented in~\cite{CW17}~\footnote{We realized that it is possible to improve the dependence to $O(m)$ in~\cite{CW17} by a more careful construction of the encoding gate in their compression scheme.}.

In this paper, we focus on time-independent Lindbladians. It is worth noting that our approach easily generalizes to time-dependent Lindbladians. We sketch this generalization in \cref{sec:timedependent}.

The rest of the paper is structured as follows. We introduce the preliminaries, including an introduction to Duhamel's principle and the algorithmic tools in \cref{sec:prelim}. In \cref{sec:series}, we present the series expansion based on Duhamel's principle and prove the error bound. In \cref{sec:quadrature}, we show how to use scaled Gaussian quadrature to approximate multiple integrals. The main theorem is proved in \cref{sec:alg} and our simulation algorithm is presented in the proof. Finally, we sketch how to generalize our method to simulating time-dependent Lindbladians in \cref{sec:timedependent}.

\section{Preliminaries}
\label{sec:prelim}
In this paper, we use $\norm{A}$ to denote the \emph{spectral norm} of a square matrix $A$, and we use $\norm{A}_1$ to denote its \emph{trace norm}. We use $I$ to denote the identity matrix and we leave its dimension implicitly when it is clear from the context. We use calligraphic font, such as $\mathcal{L}$, $\mathcal{J}$ to denote \emph{superoperators}. In particular, we use $\mathcal{I}$ to denote the \emph{identity map}. We use $\mathcal{K}[A]$ to denote the completely positive map induced by the \emph{Kraus operator} $A$, i.e.,
\begin{align}
  \mathcal{K}[A](\rho) \coloneqq A\rho A^{\dag}
\end{align}
for all $\rho$. The \emph{induced trace norm} of a superoperator $\mathcal{M}$, denoted by $\norm{\mathcal{M}}_1$, is defined as $\norm{\mathcal{M}}_1 \coloneqq \max\{\norm{\mathcal{M}(A)}_1: \norm{A} \leq 1\}$. The \emph{diamond norm} of a superoperator $\mathcal{M}$, denoted by $\norm{\mathcal{M}}_\diamond$, is defined as $\norm{\mathcal{M}}_\diamond \coloneqq \norm{\mathcal{M}\otimes \mathcal{I}}_1$, where the identity map $\mathcal{I}$ acts on a Hilbert space with the same dimension as the space $\mathcal{M}$ is acting on.

We use \emph{block-encodings} as the efficient description of the operators. Intuitively, we say a unitary $U_A$ block-encodes a matrix $A$ if $A$ appears in the upper-left block of $A$, i.e.,
\begin{align}
  U_A = 
  \begin{pmatrix}
    A/\alpha & \cdot \\
    \cdot & \cdot
  \end{pmatrix},
\end{align}
where $\alpha$ is the \emph{normalizing factor}. More formally, an $(n+b)$-qubit unitary $U_A$ is an $(\alpha, b, \epsilon)$-block-encoding of an $n$-qubit operator $A$ if
\begin{align}
  \norm{A - \alpha(\bra*{0^{\otimes b}}\otimes I)U_A(\ket*{0^{\otimes b}}\otimes I)} \leq \epsilon,
\end{align}
where the identity operator is acting on $n$ qubits.

\subsection{Duhamel's principle} 
\label{sec:Duhamel}
For a differential equation written in the form of,
\begin{equation}
    u'(t) = Lu  + f\big(t,u(t)\big), \quad u(0)= u_0,
\end{equation}
where $L$ is a linear operator, but $f$ can in principle be a nonlinear function of $u.$

The Duhamel's principle allows to separate the contribution to the solution from the initial condition and the contribution from the non-homogeneous term. Specifically, we can write the solution as 
\begin{equation}
    u=v+w,
\end{equation}
 where $v$ satisfies the equation without $f$,
\begin{equation}\label{eq: v}
    v'(t) = Lv, \quad v(0)= u_0,
\end{equation}
while $w$ follows the equation
\begin{equation}\label{eq: w}
    w'(t) = Lw  + f\big(t,u(t)\big), \quad w(0)= 0,
\end{equation}

The solution $v$, due to the fact that \cref{eq: v} is linear and homogeneous,  can be simply written as $v(t)= e^{tL}u_0.$ 
On the other hand, the equation for $w$ can be rewritten as $ \frac{\dd}{\dd t}\left(e^{-tL}w(t)\right)= e^{-tL} f\big(t,u(t)\big),$ from which a direct integration yields,
\begin{align}
w(t) = \int_0^t g(t,s)  \dd s, \quad g(t,s)\coloneqq e^{(t-s)L} f\big(s,u(s)\big).
\end{align}
Notice that when $s$ is held fixed, the function $g(t,s)$ also follows a homogeneous equation similar to \cref{eq: v},
\begin{equation}
    \frac{\dd}{\dd t} g(t,s)= L g(t,s), \quad \lim_{t\to s} g(t,s)= f\big(s,u(s)\big),
\end{equation}
which is typically how the Duhamel's principle is expressed. 

The derivation of our algorithm will heavily involve the Duhamel's principle, which can be summarized into the following formula,
\begin{equation}\label{eq: duhamel}
    u(t) = e^{tL}u_0 +  \int_0^t e^{(t-s)L} f\big(s,u(s)\big)  \dd s.
\end{equation}

\subsection{Algorithmic tools}
Given a completely positive map whose Kraus operators are given as block-encodings, we use the following lemma to probabilistically implement this complete positive map. Note that this lemma is a reformulation of the LCU for channels~\cite{CL17} in the language of block-encodings.
\begin{lemma}[Implementing completely positive maps via block-encodings of Kraus operators~\cite{LW21}]
  \label{lemma:block-encoding-channel}
  Let $A_1, \ldots, A_{m} \in \bbc^{2^n}$ be the  Kraus operators of a completely positive map. Let $U_1, \ldots, U_m \in \bbc^{2^{n+n'}}$ be their corresponding $(s_j, n', \epsilon)$-block-encodings, i.e., 
  \begin{align}
    \norm{A_j - s_j (\bra{0}\otimes I)U_j\ket{0}\otimes I)} \leq \epsilon, \quad \text{ for all $1 \leq j \leq m$}.
  \end{align}
  Let $\ket{\mu}\coloneqq \frac{1}{\sqrt{\sum_{j=1}^{m}s_j^2}}\sum_{j=1}^{m}s_j\ket{j}$. Then $(\sum_{j=1}^m\ketbra{j}{j}\otimes U_j)\ket{\mu}\ket{0}\otimes I$ implements this completely positive map in the sense that
  \begin{align}\label{eq: M-apply}
    \norm{I\otimes \bra{0}\otimes I \left(\sum_{j=1}^m\ketbra{j}{j}\otimes U_j\right) \ket{\mu}\ket{0}\ket{\psi} - \frac{1}{\sqrt{\sum_{j=1}^ms_j^2}}\sum_{j=1}^m \ket{j}A_j\ket{\psi}} \leq \frac{m\epsilon}{\sqrt{\sum_{j=1}^ms_j^2}}
  \end{align}
  for all $\ket{\psi}$.
\end{lemma}

The following lemma shows how to construct the block-encoding as a linear combination of block-encodings.
\begin{lemma}[Block-encoding of a sum of block-encodings~\cite{LW21}]
  \label{lemma:sum-to-be}
  Suppose $A \coloneqq \sum_{j=1}^m y_j A_j \in \bbc^{2^n\times 2^n}$, where $A_j \in \bbc^{2^n \times 2^n}$ and $y_j > 0$ for all $j \in \{1, \ldots m\}$. Let $U_j$ be an $(\alpha_j, a, \epsilon)$-block-encoding of $A_j$, and $B$ be a unitary acting on $b$ qubits (with $m \leq 2^b-1$) such that $B\ket{0} = \sum_{j=0}^{2^b-1}\sqrt{\alpha_jy_j/s}\ket{j}$, where $s = \sum_{j=1}^my_j\alpha_j$. Then a $(\sum_{j}y_j\alpha_j, a+b, \sum_{j}y_j\alpha_j\epsilon)$-block-encoding of $\sum_{j=1}^my_jA_j$ can be implemented with a single use of $\sum_{j=0}^{m-1}\ketbra{j}{j}\otimes U_j + ((I - \sum_{j=0}^{m-1}\ketbra{j}{j})\otimes I _{\bbc^{2^a}}\otimes I_{\bbc^{2^{n}}})$ plus twice the cost for implementing $B$.
\end{lemma}

Finally, we need the oblivious amplitude amplification for isometries.
\begin{lemma}[Oblivious amplitude amplification for isometries~{\cite{CW17}}]
  \label{lemma:aa}
  For all $a, b\in \mathbb{N}_+$, let $\ket*{\hat{0}} \coloneqq \ket*{0}^{\otimes a}$ and $\ket*{\mu}$ be arbitrary $b$-qubit state. For any $n$-qubit state $\ket*{\psi}$, let $\ket*{\hat{\psi}} \coloneqq \ket*{\hat{0}}\ket*{\mu}\ket*{\psi}$. Also define $\ket*{\hat{\phi}} \coloneqq \ket*{\hat{0}}\ket*{\phi}$, where $\ket*{\phi}$ is a $(b+n)$-qubit state. Let $P_0 \coloneqq \ketbra{\hat{0}}{\hat{0}}\otimes I_{2^b} \otimes I_{2^n}$ and $P_1 \coloneqq \ketbra{\hat{0}}{\hat{0}}\otimes \ketbra{\hat{\mu}}{\hat{\mu}} \otimes I_{2^n}$ be two projectors. Suppose there exists an operator $W$ satisfying
  \begin{align}
    W\ket*{\hat{\phi}} = \frac{1}{2}\ket*{\hat{\phi}} + \sqrt{\frac{3}{4}}\ket*{\hat{\phi}^{\bot}},
  \end{align}
  where $\ket*{\hat{\phi}}$ satisfies $P_0\ket*{\hat{\phi}^{\bot}} = 0$. Then it holds that
  \begin{align}
    -W(I-2P_1)W^{\dag}(I-2P_0)W\ket*{\hat{\psi}} = \ket*{\hat{\phi}}.
  \end{align}
\end{lemma}

\section{Higher order series expansion based on Duhamel's principle}
\label{sec:series}

The goal of our quantum algorithm is to simulate the Lindblad equation defined in \cref{eq: lindb}. In the context of quantum trajectory theory \cite{plenio1998quantum}, we view the commutator and the anti-commutator terms as the \emph{drifting part}, and  the $L_j\rho L_j^{\dag}$ terms are regarded as \emph{jump part}.  Accordingly, motivated by the numerical method in~\cite{CL21}. We decompose $\mathcal{L}$ into two superoperators, the \emph{drifting part} $\mathcal{L}_{\mathrm{D}}$ and the \emph{jump part} $\mathcal{L}_{\mathrm{J}}$. Namely,
\begin{align}\label{eq: LLJLL}
  \mathcal{L} = \mathcal{L}_\mathrm{D} &+ \mathcal{L}_\mathrm{J}, \quad \text{ and}\\
  \mathcal{L}_\mathrm{D}(\rho) \coloneqq J\rho + \rho J^{\dag}, &\quad \mathcal{L}_\mathrm{J}(\rho) \coloneqq \sum_{j=1}^m L_j\rho L_j^{\dag}.
\end{align}
Here  we define $J$ as
\begin{align}
  J\coloneqq -i H_{\mathrm{eff}}. 
\end{align}
where an \emph{effective Hamiltonian} $H_{\mathrm{eff}}$ is given by,
\begin{align}
  H_{\mathrm{eff}} \coloneqq H + \frac{1}{2i}\sum_{j=1}^m L_j^{\dag}L_j.
\end{align}
Thus $\mathcal{L}_\mathrm{D}$ can be viewed as a generalized anti-commutator.

By treating the term with $\mathcal{L}_\mathrm{D}$ as a non-homogeneous term,  the Duhamel's principle in \cref{eq: duhamel} can be applied, and we get,
\begin{align}\label{eq: formula}
  \rho_t = e^{\mathcal{L}t}(\rho_0) = e^{\mathcal{L}_\mathrm{D}t}(\rho_0) + \int_0^t e^{\mathcal{L}_\mathrm{D}(t-s)}(\mathcal{L}_\mathrm{J} \rho_s) \dd s.
\end{align}

Note that the solution $\rho_s$ is still involved in the integral on the right hand side. Therefore, this equation does not provide an explicit formula for the solution; Rather, it offers an integral representation of the Lindblad equation. Nevertheless, one can apply \cref{eq: duhamel}  again to $\rho_s$ in the integral.   
After $K$ such iterations, we arrive at
\begin{align}
  \begin{aligned}\label{eq: duhem}
    \rho_t &= e^{\mathcal{L}_\mathrm{D}t}(\rho_0) \\
             &+\sum_{k=1}^K \int_{0 \leq s_1 \leq \cdots \leq s_k \leq t} e^{\mathcal{L}_\mathrm{D}(t-s_k)} \mathcal{L}_\mathrm{J} e^{\mathcal{L}_\mathrm{D}(s_k-s_{k-1})} \mathcal{L}_\mathrm{J} \cdots e^{\mathcal{L}_\mathrm{D}(s_2-s_1)} \mathcal{L}_\mathrm{J} e^{\mathcal{L}_\mathrm{D}(s_1)} (\rho_0) \dd s_1 \cdots \dd s_k \\
             &+\int_{0 \leq s_1 \leq \cdots \leq s_{K+1} \leq t} e^{\mathcal{L}_\mathrm{D}(t-s_{K+1})} \mathcal{L}_\mathrm{J} e^{\mathcal{L}_\mathrm{D}(s_{K+1}-s_K)} \mathcal{L}_\mathrm{J} \cdots e^{\mathcal{L}_\mathrm{D}(s_2-s_1)}  \mathcal{L}_\mathrm{J}(\rho_{s_1}) \dd s_1 \cdots \dd s_{K+1}.
    \end{aligned}
\end{align}
Now notice that the first two terms on the right hand side only depend on the initial density matrix, and thus they are amenable to numerical approximations. 
Meanwhile, the last term will be regarded as a truncation error, which will be bounded later. 

We first derive the Kraus representation of $e^{\mathcal{L}_\mathrm{D}t}$, which is the first term in the expansion, but also appears in each integral.
The Kraus form can be obtained from a Taylor series. To see this, let us consider the case where $\mathcal{L}_\mathrm{D}$ is acting on a pure state $\ket{\psi}$:
\begin{align}
  \label{eq:scho}
  \frac{\dd}{\dd t}\ket{\psi} = J\ket{\psi}.
\end{align}
Using the chain rule, we have
\begin{align}
  \label{eq:vonneu}
  \frac{\dd}{\dd t}\ketbra{\psi}{\psi} = J\ketbra{\psi}{\psi} + \ketbra{\psi}{\psi}J^{\dag} = \mathcal{L}_\mathrm{D}(\ketbra{\psi}{\psi}).
\end{align}
The above two equations also hold for general states $\rho$ by linearity. Hence, solving \cref{eq:scho} and \cref{eq:vonneu} yields
\begin{align}
  e^{\mathcal{L}_\mathrm{D}t} = \mathcal{K}\left[e^{tJ}\right].
\end{align}

Now, we adopt the notation from \cite{CL21},
\begin{align}
  \label{eq:fk}
   \mathcal{F}_k(s_k, \ldots, s_1) \coloneqq
   \mathcal{K}[e^{J(t-s_k)}] \mathcal{L}_\mathrm{J} \mathcal{K}[e^{J(s_k-s_{k-1})}] \mathcal{L}_\mathrm{J} \cdots \mathcal{K}[e^{J(s_2-s_1)}] \mathcal{L}_\mathrm{J} \mathcal{K}[e^{J(s_1-0)}].
\end{align}
  
We further define
\begin{align}\label{eq: GK}
  \mathcal{G}_{K}(t) \coloneqq \mathcal{K}[e^{Jt}] + \sum_{k=1}^K \int_{0 \leq s_1 \leq \cdots \leq s_k \leq t} \mathcal{F}_{k}(s_k, \ldots, s_1)\, \dd s_1 \cdots \dd s_k.
\end{align}
At this point, the problem is reduced to approximating $e^{\mathcal{L}t}$ by $\mathcal{G}_M(t)$. We first prove an error bound.
\begin{restatable}{lemma}{higherorder}
  \label{lemma:gk}
  It holds that
  \begin{align}
    \norm{e^{\mathcal{L}t} - \mathcal{G}_K(t)}_{\diamond} \leq \frac{(2\norm{\mathcal{L}}_{\mathrm{be}} t)^{K+1}}{(K+1)!}.
  \end{align}
\end{restatable}
The proof of \cref{lemma:gk} is postponed to \cref{sec:proof-higherorder}.

Eventually, we need to approximate the Kraus operator $e^{Jt}$ in our quantum algorithm. This can be done by a truncated Taylor series. For notational simplicity, we define
\begin{align}\label{eq: JJK}
  \quad \mathcal{J}_{K'} = \mathcal{K}\left[\sum_{\ell=0}^{K'} \frac{J^\ell t^\ell}{\ell!} \right].
\end{align}
We quantify the error of this approximation in the following lemma.

\begin{restatable}{lemma}{jktaylor}
  \label{lem:taylor}
  Suppose that $k \in \mathbb{N}$ such that $(k+1)! \geq 2\norm{J}^{k+1}t^{k+1}$. 
Let $\mathcal{J}_k$ be defined in \cref{eq: JJK}. 
  It holds that
  \begin{align}
    \norm{\mathcal{K}[e^{Jt}] - \mathcal{J}_k(t)}_{\diamond} \leq \frac{8e^{\norm{\mathcal{L}}_\mathrm{be}t}\norm{\mathcal{L}}_{\mathrm{be}}^{k+1}t^{k+1}}{(k+1)!}.
  \end{align}
\end{restatable}
The proof of \cref{lem:taylor} is postponed to \cref{sec:proof-higherorder}.

We also provide the following useful lemma, which will be used in the final analysis of our algorithm
\begin{restatable}{lemma}{diamondnorm}
  \label{lemma:jk-fk}
  Suppose that $k,K' \in \mathbb{N}$ such that $(K'+1)! \geq 8e^{\norm{\mathcal{L}}_{\mathrm{be}}t}\norm{\mathcal{L}}_{\mathrm{be}}^{K'+1}t^{K'+1}$. Let $\mathcal{J}_{K'}$ be defined in  \cref{eq: JJK}, and $\mathcal{F}_k$ be defined in \cref{eq:fk}. 
  It holds that
  \begin{align}
    \begin{aligned}
      \|\mathcal{J}_{K'}(t-s_m) &\mathcal{L}_\mathrm{J} \mathcal{J}_{K'}(s_m-s_{m-1}) \mathcal{L}_\mathrm{J} \cdots \mathcal{J}_{K'}(s_2-s_1) \mathcal{L}_\mathrm{J} \mathcal{J}_{K'}(s_1)  -\mathcal{F}_{k}(s_k, \ldots, s_1)\|_{\diamond} \\
                                  &\leq \frac{8e^{\norm{\mathcal{L}}_{\mathrm{be}}t}\norm{\mathcal{L}}_{\mathrm{be}}^{K'+1}}{(K'+1)!}(2\norm{\mathcal{L}}_{\mathrm{be}})^k2^kt^{K'+1}. 
    \end{aligned}
  \end{align}
\end{restatable}
The proof of \cref{lemma:jk-fk} is postponed to \cref{sec:proof-higherorder}.

\section{Approximating multiple integrals using scaled Gaussian quadrature}
\label{sec:quadrature}
To obtain an algorithm that can be directly implemented, we apply Gaussian quadrature formulas to approximate the multiple integrals in \cref{eq: GK}. Due to their optimal accuracy, the number of terms in the approximation is significantly compressed. Typically, quadrature error depends on the smoothness of the function. For this purpose, 
we first bound the derivatives of $\mathcal{F}_{k}$.
\begin{restatable}{lemma}{derivativebound}
  \label{lemma:derivative}
  For all $k' \in [k]$, it holds that
  \begin{align}
    \norm{\frac{\dd^{k'}}{\dd s_j^{k'}}\mathcal{F}_{k}}_{\diamond} \leq 2^{2k'+k}\norm{\mathcal{L}}_{\mathrm{be}}^k\norm{J}^{k'}.
  \end{align}
\end{restatable}
The proof of \cref{lemma:derivative} is postponed to \cref{sec:proof-quadrature}.

We now discuss the quadrature approximation for the integral in \cref{eq: GK}:
\begin{align}\label{multi-int}
  \int_{0 \leq s_1 \leq \cdots \leq s_k \leq t}\mathcal{L}_\mathrm{J} \mathcal{K}[e^{J(s_m-s_{m-1})}] \mathcal{L}_\mathrm{J} \cdots \mathcal{K}[e^{J(s_2-s_1)}] \mathcal{L}_\mathrm{J} \mathcal{K}[e^{J(s_1-0)}]\,\dd s_1\cdots \dd s_k.
\end{align}

For optimal accuracy, 
we use Gaussian quadrature.
In the Gaussian quadrature rule, the interpolation nodes $0 \leq \hat{s}_1 \leq \cdots \leq \hat{s}_q \leq t$ and the weights $w_1, \ldots, w_q$ are independent of the function and can be pre-computed. More specifically, let $\{\mathscr{P}_i(x)\}_{i}$ be the standard Legendre polynomials. They are an orthonormal family of polynomials in the sense that
\begin{align}
  \int_{-1}^1\mathscr{P}_r(x)\mathscr{P}_s(x)\,\dd x = 
  \begin{cases}
    0 & \quad r \neq s,\\
    1 & \quad r = s.
  \end{cases}
\end{align}

By a simple scaling, 
\begin{align}
  \hat{\mathscr{P}}(x) \coloneqq \mathscr{P}\left(\frac{2x}{t}-1\right).
\end{align}
we obtain the functions $\hat{\mathscr{P}}$ that are orthogonal for the interval $[0, t]$.
Let $\{\hat{s}_i\}_{i=1}^n$ be the roots of the $n$-th degree polynomial $\hat{\mathscr{P}}_q$. We also define
\begin{align}
  \pi_q(x) \coloneqq (x-\hat{s}_1)(x-\hat{s}_2)\cdots(x-\hat{s}_q).
\end{align}
Then the $i$-th Lagrange polynomial for points $\hat{s}_1,\ldots,\hat{s}_q$ is
\begin{align}
  \mathscr{L}_{q-1,i}(x) = \frac{\pi_q(x)}{(x-\hat{s}_i)\pi_q'(\hat{s}_i)}.
\end{align}
Once the quadrature points are selected, the weight of the Gaussian quadrature can be expressed as
\begin{align}
  w_i = \int_0^t\mathscr{L}_{q-1,i}(x)\,\dd x,
\end{align}
which can be direct deduced from a polynomial interpolation. 

We refer to $\hat{s}_1,\ldots, \hat{s}_q$ as the \emph{canonical quadrature points} and $w_1, \ldots, w_q$ as the \emph{canonical weights}. In  approximating $\int_0^t f(x)\,\dd x$ using $\sum_{j=1}^q f(\hat{s}_j)w_j$, the error follows the standard bound,
\begin{equation}\label{eq: GQ}
    E_q[f]= \int_0^t f(x)\,\dd x -  \sum_{j=1}^q f(\hat{s}_j)w_j =
    \frac{f^{(2q)}(\xi)}{(2q)!}\int_0^t \pi_q(x)^2 \dd x,
\end{equation}
for some $\xi \in [0,t].$

To estimate the integral term in the error, we notice that, 
\begin{equation}
  \pi_q(x)= \frac{t^q q! }{2^q (2q-1)!!} \mathscr{P}_q\left(\frac{2x}{t}-1\right),    
\end{equation}
The coefficient on the right hand side is determined by observing that $\pi_q(x)$ is a monic polynomial, and the leading coefficient of the standard Legendre polynomial is $(2q-1)!!/q!$.
The notation $!!$ indicates a double factorial, i.e., $(2q-1)!!=(2q-1)(2q-3)\cdots 1$ and we use the convention that $(-1)!!=1.$ 

Combining these formulas, we arrive at an explicit bound
\begin{align}\label{eq: gq-err}
   \abs{ E_q [f]}= \frac{\abs{f^{(2q)}}(\xi)t^{2q+1} (q!)^2 }{(2q)! 2^{2q} (2q-1)!! (2q+1)!! } \leq \frac{\abs{f^{(2q)}}(\xi)t^{2q+1} q }{(2q)! 2^{4q-1} }
\end{align}
for some $\xi \in [0, t]$, where the inequality follows from the fact that $q!! = 2^{q/2}(q/2)!$ for even $q$. Here we also used the identity,
\[ \int_{-1}^1 \mathscr{P}_q^2(x) \,\dd x = \frac{2}{2q+1}.
\]

We hold $t$ fixed. For an interval $[0, s_k]$ with $s_k \leq t$, we use a scaled canonical quadrature points and weights: $s_k\hat{s_1}/t, \ldots, s_k\hat{s}_q/t$, and $s_kw_1/t, \ldots, s_kw_q/t$. Then, $\int_0^{s_k}f(x)\,\dd x$ can be approximated by the scaled quadrature points and weights with the same error bound:
\begin{equation}
    \int_0^{s_k}f(x)\,\dd x = \sum_j^n f\left( \frac{s_k \hat{s}_j }{t}\right) \frac{s_k w_j}{t}+ \mathcal{O}\left( \frac{\norm{f^{(2n)}}_\infty s_k^{2n+1} n }{(2n)! 2^{4n-1} } \right).  
\end{equation}

For each $j \in [n]$, define the functions $u_k$ and $v_k$ as
\begin{align}\label{eq: uv-scale}
    u_j(x) &\coloneqq x \hat{s}_j/t \\
    v_j(x) &\coloneqq x w_j/t.
\end{align}
Note that for any $s_\ell$, $\{u_j(s_\ell)\}_{j=1}^q$ are the scaled canonical quadrature point and $\{v_j(s_\ell)\}_{j=1}^q$ are the scaled weights.

To simplify the notation, we extend \cref{eq: uv-scale} and  define
\begin{align}
  \hat{x}_{(j_k)} \coloneqq \hat{s}_{j_k}, \quad &\text{ and } \quad \hat{x}_{(j_{k},\ldots, j_{k-\ell})} \coloneqq u_{j_{k-\ell}}\circ\cdots\circ u_{j_{k-1}} (\hat{s}_{j_k}) \, \text{ for all $1 \leq \ell \leq k-1$},\\
  \hat{w}_{(j_k)} \coloneqq w_{j_k}, \quad &\text{ and } \quad \hat{w}_{(j_k,\ldots, j_{k-\ell})} \coloneqq v_{j_{k-\ell}} (\hat{x}_{(j_k,\ldots,j_{k-\ell+1})}) \, \text{ for all $1 \leq \ell \leq k-1$}.
\end{align}

With these notations, the approximation of the integral in \cref{multi-int} becomes,
\begin{align}
  \sum_{j_1=1}^{q} \cdots \sum_{j_k=1}^{q}\mathcal{F}_{k}\left(\hat{x}_{(j_k)}, , \ldots, \hat{x}_{(j_k, \ldots, j_1)}\right)\hat{w}_{(j_k)}\cdots \hat{w}_{(j_k, \ldots, j_1)}.
\end{align}

 We first show some useful properties of the quadrature weights.
 \begin{restatable}{lemma}{scaledweight}
   \label{lemma:wsl}
  for all $\ell \in \{0, \ldots, q\}$, it holds that
  \begin{align}
    \label{eq:wsl}
    \sum_{j=1}^q w_{j}\hat{s_j}^\ell  = \frac{t^{\ell+1}}{\ell+1}.
  \end{align}
  In particular, when $\ell=0$
  \begin{align}
    \label{eq:sum-w}
    \sum_{j=1}^q w_{j}  = t.
  \end{align}
  For all positive integers $k,\ell$ with $\ell < k$, it also holds that
  \begin{align}
    \label{eq:ww}
    \sum_{j_k=1}^q\cdots\sum_{j_{k-\ell}=1}^q\hat{w}_{(j_k)}\cdots\hat{w}_{(j_k,\ldots,j_{k-\ell})} = \frac{t^{\ell+1}}{(\ell+1)!}.
  \end{align}
  In particular, when $\ell = k-1$, it holds that
  \begin{align}
    \label{eq:wwkto1}
    \sum_{j_k=1}^q\cdots\sum_{j_{1}=1}^q\hat{w}_{(j_k)}\cdots\hat{w}_{(j_k,\ldots,j_{1})} = \frac{t^{k}}{(k+1)!}.
  \end{align}
\end{restatable}
The proof of \cref{lemma:wsl} is postponed to \cref{sec:proof-quadrature}.

With the bound on the derivatives of the integrand in \cref{lemma:derivative} and the Gaussian quadrature error \cref{eq: gq-err}, we can estimate the overall quadrature error, as stated in the following lemma,
\begin{restatable}{lemma}{quadratureerror}
  \label{lemma:int-to-sum}
  It holds that
\begin{align}
  &\begin{aligned}
  &\left\|\int_{0 \leq s_1 \leq \cdots \leq s_k \leq t}\mathcal{F}_{k}(s_k, \ldots, s_1) \,\dd s_1\cdots \dd s_k \right. \\
  &\left.\quad-\sum_{j_1=1}^{q} \cdots \sum_{j_k=1}^{q}\mathcal{F}_{k}\left(\hat{x}_{(j_k)}, , \ldots, \hat{x}_{(j_k, \ldots, j_1)}\right)\hat{w}_{(j_k)}\cdots \hat{w}_{(j_k, \ldots, j_1)}\right\|_{\diamond} \\
    &\quad = O\left(\frac{(2t)^{k-1}2^{k+1}\norm{\mathcal{L}}_{\mathrm{be}}^k\norm{J}^{(2q)}t^{2q+1}q}{(k-1)!(2q)!}\right).
  \end{aligned}
\end{align}
\end{restatable}
The proof of \cref{lemma:int-to-sum} is postponed to \cref{sec:proof-quadrature}.

\section{Quantum algorithm and the proof of the main theorem}
\label{sec:alg}

In this section, we prove the main theorem and describe the algorithm in the proof. Our algorithm constructs a segment for constant evolution time, i.e., $t\norm{\mathcal{L}}_{\mathrm{be}} = \Theta(1)$. For arbitrary evolution time, we just repeat the simulation segment $O(t\norm{\mathcal{L}}_{\mathrm{be}})$ times with a scaled precision parameter.

We first present the higher order approximation of $e^{\mathcal{L}}$ as a completely positive map with explicit Kraus operators. Then we use \cref{lemma:block-encoding-channel} to implement this completely positive map with success probability $1/4$, which can be calculated by analyzing the normalizing constants of the block-encodings of the Kraus operators. Then we show that the special state $\ket{\mu}$ required by \cref{lemma:block-encoding-channel} can be efficiently prepared. Finally, we analyze the error introduced by using a truncated Taylor series to approximate $e^{Jt}$, which is part of the Kraus operators.

\begin{theorem}
  \label{thm:main}
  Suppose we are given an $(\alpha_0, a, \epsilon')$-block-encoding $U_H$ of $H$, and an $(\alpha_j, a, \epsilon')$-block-encoding $U_{L_j}$ for each $L_j$. For all $t,\epsilon \geq 0$ with $\epsilon' \leq \epsilon/(t(\norm{\mathcal{L}}_{\mathrm{be}})$, there exists a quantum algorithm for simulating $e^{\mathcal{L}t}$ using
  \begin{align}
    O\left(t\norm{\mathcal{L}}_{\mathrm{be}}\,\frac{\log(t\norm{\mathcal{L}}_{\mathrm{be}}/\epsilon)}{\log\log(t\norm{\mathcal{L}}_{\mathrm{be}}/\epsilon)}\right)
  \end{align}
  queries to $U_H$ and $U_{L_j}$ and 
  \begin{align}
    O\left(mt\norm{\mathcal{L}}_{\mathrm{be}}\,\left(\frac{\log(t\norm{\mathcal{L}}_{\mathrm{be}}/\epsilon)}{\log\log(t\norm{\mathcal{L}}_{\mathrm{be}}/\epsilon)}\right)^2\right)
  \end{align}
  additional 1- and 2-qubit gates.
\end{theorem}
\begin{proof}
  We describe our simulation algorithm and prove the theorem as follows.

  \textbf{Completely-positive approximation.}
  The final approximation to $e^{\mathcal{L}t}$ is
  \begin{align}
    \label{eq:final-approx}
    \mathcal{K}[e^{Jt}] +\sum_{k=1}^K \sum_{j_1=1}^{q} \cdots \sum_{j_k=1}^{q}\mathcal{F}_{K}\left(\hat{x}_{(j_k)}, , \ldots, \hat{x}_{(j_k, \ldots, j_1)}\right)\hat{w}_{(j_k)}\cdots \hat{w}_{(j_k, \ldots, j_1)},
  \end{align}
  which is a completely positive map and the block-encodings of its Kraus operators can be easily obtained by a product of the block-encoding $U_H$ of $H$, and the block-encodings $U_{L_j}$ of $L_j$ as well as positive factors determined by Taylor's expansion and the Gaussian quadrature weights. More specifically, define the index sets $\mathscr{I}$ as
  \begin{align}
    \mathscr{I} &\coloneqq \left\{k,\ell_1, \ldots, \ell_{k}, j_1, \ldots, j_k : k \in [K], \ell_1,\ldots,\ell_k \in [m], j_1\ldots, j_k \in [q]\right\}.
  \end{align}
  The completely positive map in \cref{eq:final-approx} can be written as
  \begin{align}
    A_0\rho A_0^{\dag} + \sum_{j\in \mathscr{I}}A_{j}\rho A_{j}^{\dag},
  \end{align}
  with $A_0 \coloneqq e^{Jt}$,
  and
  \begin{align}
    \label{eq:aj}
    A_j = \sqrt{\hat{w}_{(j_k)}\cdots\hat{w}_{(j_k,\ldots,j_1)}} e^{J(t-\hat{x}_{(j_k)})}L_{\ell_{k}}\cdots e^{J(\hat{x}_{(j_k,\ldots,j_{2})}-\hat{x}_{(j_k,\ldots,j_{1})})} L_{\ell_1}e^{J(\hat{x}_{(j_k,\ldots,j_{1})})},
  \end{align}
  for $j = (k, \ell_1,\ldots, \ell_{k-1},j_1,\ldots,j_k) \in \mathscr{I}$.

  \textbf{Setting parameters for 1/4 success probability.}
  We use \cref{lemma:block-encoding-channel} to implement the above map, and the success probability is determined by the sum-of-squares of the normalizing constants of the Kraus operators.
  
  We first consider the Kraus operators of
  \begin{align}
     \mathcal{F}_{k}(s_k, \ldots, s_1) =
     \mathcal{K}[e^{J(t-s_k)} \mathcal{L}_\mathrm{J} \mathcal{K}[e^{J(s_k-s_{k-1})}] \mathcal{L}_\mathrm{J} \cdots \mathcal{K}[e^{J(s_2-s_1)}] \mathcal{L}_\mathrm{J} \mathcal{K}[e^{J(s_1-0)}],
  \end{align}
  for any $s_1\leq\cdots\leq s_k$. For each $\mathcal{K}[e^{Js}]$, we use \cref{lemma:sum-to-be} to approximate its block-encoding as a truncated Taylor series. For the convenience of analysis, let us for now assume an infinite Taylor series is implemented. The normalizing constant of the block-encoding for $e^{Js}$ is then
  \begin{align}
    \sum_{\ell=0}^{\infty}\frac{s^{\ell}(\alpha_0 + \frac{1}{2}\sum_{j=1}^m\alpha_j^2)}{\ell!} = e^{s\norm{\mathcal{L}}_{\mathrm{be}}}.
  \end{align}
  As a result, the sum-of-squares of the normalizing constants of the Kraus operators of $\mathcal{F}_k(s_k, \ldots, s_1)$ is
  \begin{align}
    \sum_{j_1,\ldots,j_k=0}^m e^{2\norm{L}_{\mathrm{be}}(t-s_k)} e^{2\norm{L}_{\mathrm{be}}(s_k-s_{k-1})}\cdots e^{2\norm{L}_{\mathrm{be}}(s_1-0)} \alpha_{j_1}^2\cdots\alpha_{j_k}^2= e^{2\norm{L}_{\mathrm{be}}t} \left(\sum_{j=1}^m\alpha_j^2\right)^k
  \end{align}
  For the approximation of the integral, 
  \begin{align}
    \sum_{j_1=1}^{q} \cdots \sum_{j_k=1}^{q}\mathcal{F}_{k}\left(\hat{x}_{(j_k)}, , \ldots, \hat{x}_{(j_k, \ldots, j_1)}\right)\hat{w}_{(j_k)}\cdots \hat{w}_{(j_k, \ldots, j_1)},
  \end{align}
  the sum-of-squares of the normalizing constants of its Kraus operators is
  \begin{align}
    e^{2\norm{\mathcal{L}}_{\mathrm{be}}t}\left(\sum_{j=1}^m\alpha_j^2\right)^k \sum_{j_1=1}^{q} \cdots \sum_{j_k=1}^{q}\hat{w}_{(j_k)}\cdots \hat{w}_{(j_k, \ldots, j_1)} = \frac{t^k}{(k-1)!}e^{2\norm{\mathcal{L}}_{\mathrm{be}}t}\left(\sum_{j=1}^m\alpha_j^2\right)^k.
  \end{align}
  Finally, for the approximation
  \begin{align}
    \mathcal{K}[e^{t}]
    +\sum_{k=1}^K \sum_{j_1=1}^{q} \cdots \sum_{j_k=1}^{q}\mathcal{F}_{k}\left(\hat{x}_{(j_k)}, , \ldots, \hat{x}_{(j_k, \ldots, j_1)}\right)\hat{w}_{(j_k)}\cdots \hat{w}_{(j_k, \ldots, j_1)},
  \end{align}
  the sum-of-squares of the normalizing constants of its Kraus operators is
  \begin{align}
    e^{2\norm{\mathcal{L}}_{\mathrm{be}}t} &+ \sum_{k=1}^K \frac{t^k}{(k-1)!}e^{2\norm{\mathcal{L}}_{\mathrm{be}}t}\left(\sum_{j=1}^m\alpha_j^2\right)^k \\
    &= e^{2\norm{\mathcal{L}}_{\mathrm{be}}t} + te^{2\norm{\mathcal{L}}_{\mathrm{be}}t}\sum_j\alpha_j^2\sum_{k=1}^K \frac{t^{k-1}(\sum_j\alpha_j^2)^{k-1}}{(k-1)!} \\
                                                                                                                                                        &\leq e^{2\norm{\mathcal{L}}_{\mathrm{be}}t} + t\sum_j\alpha_j^2e^{2\norm{\mathcal{L}}_{\mathrm{be}}t}e^{t\sum_j\alpha_j^2}.
  \end{align}

  Note that the inequality above provides an upper bound for the sum-of-squares of the normalizing constants. There is a closed-form expression for this quantity. By setting the right hand side to 2, and  solve the equation, the above upper bound implies that $t$ must satisfy
  \begin{align}
    t\norm{\mathcal{L}}_{\mathrm{be}} = \Theta(1).
  \end{align}
  Then we use \cref{lemma:aa} to boost the success probability to 1 with only three application of the circuit.
  
  \textbf{Determining truncation orders.}
  Now, we analyze the error by setting $t\norm{\mathcal{L}}_{\mathrm{be}} = \Theta(1)$. By \cref{lemma:gk} and \cref{lemma:int-to-sum}, the total approximation error can be bounded by the following.
  \begin{align}
    &\norm{e^{\mathcal{L}t} - \mathcal{K}[e^{Jt}] -\sum_{k=1}^K \sum_{j_1=1}^{q} \cdots \sum_{j_k=1}^{q}\mathcal{F}_{k}\left(\hat{x}_{(j_k)}, , \ldots, \hat{x}_{(j_k, \ldots, j_1)}\right)\hat{w}_{(j_k)}\cdots \hat{w}_{(j_k, \ldots, j_1)}}_{\diamond}  \\
    &\leq \frac{(2\norm{\mathcal{L}}_{\mathrm{be}})^{K+1}}{(K+1)!}  + O\left(\sum_{k=1}^K\frac{(2t)^{k-1}2^{k+1}\norm{\mathcal{L}}_{\mathrm{be}}^k\norm{J}^{(2q)}t^{2q+1}q}{(k-1)!(2q)!}\right) \\
    &= \frac{(2\norm{\mathcal{L}}_{\mathrm{be}})^{K+1}}{(K+1)!} + \frac{\norm{J}^{2q}t^{2q+1}q}{(2q)!}O\left(\sum_{k=1}^K\frac{(2t)^{k-1}2^{k+1}\norm{\mathcal{L}}_{\mathrm{be}}^k}{(k-1)!}\right) \\
    &= \frac{(2\norm{\mathcal{L}}_{\mathrm{be}})^{K+1}}{(K+1)!} + \frac{\norm{J}^{2q}t^{2q+1}q}{(2q)!}O\left(e^{4t\norm{\mathcal{L}}_{\mathrm{be}}}\right) \\
    &\leq \frac{(2\norm{\mathcal{L}}_{\mathrm{be}})^{K+1}}{(K+1)!} + \frac{\norm{\mathcal{L}}_{\mathrm{be}}^{2q}t^{2q+1}q}{(2q)!}O\left(e^{4t\norm{\mathcal{L}}_{\mathrm{be}}}\right),
  \end{align}
  where the last inequality follows from
  \begin{align}
    \norm{J} \leq \norm{H} + \frac{1}{2}\sum_{j}\norm{L_j}^2 \leq \alpha_0 + \frac{1}{2}\sum_j\alpha_j^2 = \norm{\mathcal{L}}_{\mathrm{be}}.
  \end{align}
  With $t\norm{\mathcal{L}}_{\mathrm{be}} = \Theta(1)$, it suffices to set
  \begin{align}
    K,q = O\left(\frac{\log(1/\epsilon)}{\log\log(1/\epsilon)}\right)
  \end{align}
  to make the approximation error at most $\epsilon/2$.

  \textbf{Applying the main algorithmic tool (\cref{lemma:block-encoding-channel}).}
  In \cref{lemma:block-encoding-channel}, we need to prepare the special state $\ket{\mu}$ which encodes a superposition of normalizing constants of all Kraus operators. Now we show how to efficiently prepare this state. First observe that the normalizing constant for $A_j$ in \cref{eq:aj} is
  \begin{align}
    &\quad \sqrt{\hat{w}_{(j_k)}\cdots\hat{w}_{(j_k,\ldots,j_1)}} e^{\norm{\mathcal{L}}_{\mathrm{be}}(t-\hat{x}_{(j_k)})}\alpha_{\ell_{k}}\cdots e^{\norm{\mathcal{L}}_{\mathrm{be}}(\hat{x}_{(j_k,\ldots,j_{2})}-\hat{x}_{(j_k,\ldots,j_{1})})} \alpha_{\ell_1}e^{\norm{\mathcal{L}}_{\mathrm{be}}(\hat{x}_{(j_k,\ldots,j_{1})})} \\
    &=\sqrt{\hat{w}_{(j_k)}\cdots\hat{w}_{(j_k,\ldots,j_1)}}e^{\norm{\mathcal{L}}_{\mathrm{be}}t}\alpha_{\ell_k}\cdots\alpha_{\ell_1} \\
    &= \frac{1}{\sqrt{t^{k(k-1)/2}}}\sqrt{w_{j_k}\hat{s}_{j_k}^{k-1}w_{j_{k-1}}\hat{s}_{j_{k-1}}^{k-2}\cdots w_{j_2}\hat{s}_{j_2} w_{j_1}}e^{\norm{\mathcal{L}}_{\mathrm{be}}t}\alpha_{\ell_k}\cdots\alpha_{\ell_1},
  \end{align}
  and the normalizing constant for $A_0$ is $e^{\norm{\mathcal{L}}_{\mathrm{be}}t}$. Note that $e^{\norm{\mathcal{L}}_{\mathrm{be}}t}$ appears in every amplitude of $\ket{\mu}$ and therefore can be ignored. We use three registers in $\ket{\mu}$. The first register contains $K$ qubits and it encodes $k$ in unary representation, i.e., we use $\ket{1^k0^{K-1}}$ to represent $k$. The second register contains $k$ subregisters of $\log m$ qubits to represent $\ell_1,\ldots, \ell_k$. The third register contains $k$ subregisters of $\log q$ qubits to represent $j_1,\ldots, j_k$. We first prepare the normalized version of the  state $\sum_{k=0}^K\frac{1}{\sqrt{t^{k(k-1)/2}}}\ket{1^k0^{K-1}}$, which can be done using $O(K)$ gates: we apply a rotation on the first qubit, and then apply a rotation on each subsequent qubit controlled by the previous qubit. For the second register, in each subregister we prepare the normalized version of $\sum_{j=1}^m\alpha_j\ket{j}$. The total gate cost for the second register is $O(Km)$. For the $\ell$-th subregister of the third register, we prepare the normalized version of $\sum_{j=1}^q\sqrt{w_js_j^{\ell-1}\ket{j}}$. The total gate cost for the third register is $O(Kq)$. Note that each gate acting on the $\ell$-th subregister of the second and the third register is further controlled on the $\ell$-th qubit of the first register, which effects the truncation. Therefore, the total gate cost for preparing $\ket{\mu}$ is $O(K(m+q))$.

  Now, we show how to use \cref{lemma:sum-to-be} to approximate the block-encoding $U_s$ of $e^{Js}$ for any $0 < s \leq t$, where $s$ is provided in a time register containing $\ket{s}$. Here we use $K'$ to denote the Taylor series truncation error. So we need to use \cref{lemma:sum-to-be} to implement a block-encoding of $\sum_{k=0}^{K'}s^kJ^k/k!$. Recall that $J = -iH - \frac{1}{2}\sum_{j=1}^mL_j^{\dag}L_j$. In \cref{lemma:sum-to-be}, we need to implement the $B$ gate for preparing a superposition of coefficients. We use $K'+1$ control registers: the first register contains $K$ qubits which encode $k$ in unary; each subsequent register contains $O(\log(m))$ qubits. The $B$ gate is implemented as follows. Controlled by the time register $\ket{s}$, we implement the normalized version of the state $\sum_{k=0}^{K'}\sqrt{s^k/k!}\ket{1^k0^{K'-k}}$ on the first control register, which can be done with $O(K)$ controlled-rotations. For each subsequent control resister, we implement the normalized version of the state $\ket{0}+\sum_{j=1}^m\sqrt{\alpha_j^2/2}\ket{j}$, which costs $O(m)$ gates. The controlled operation $\sum_j\ketbra{j}{j}\otimes A_j$ can be implemented by the controlled-$U_H$ and controlled-$U_{L_j}$ controlled by the $K+1$ control registers. Therefore, the total gate cost for implementing $B$ is $O(K'm)$. The controlled rotations on the first control register controlled by the time register costs $O(\poly(b))$ gates where $b$ is the bits used to represent $s$. It suffices to set $b = O(\log(1/\epsilon))$ for a precise representation of $s$ within $\epsilon$. As a result, the cost $O(\poly(b))$ is not dominating. As a result, the total gate cost for implementing $\sum_s \ketbra{s}{s}\otimes U_{s}$ is $O(Km)$. 

  \textbf{Additional approximation.}
  It is important to note that by a direct application of \cref{lemma:sum-to-be}, the error of the block-encoding we implement is $O(\epsilon' e^{s\norm{\mathcal{L}}_{\mathrm{be}}})$. However, a more careful analysis shows a much better error bound: first assume we had implemented the infinity Taylor series. Then the error $\epsilon'$ of each block-encoding will cause error for the implementation that is bounded by\footnote{The inequality $\norm*{e^{Js} - e^{\tilde{J}s}} \leq \norm{J - \tilde J}s$ does not hold for general matrices $J$. However, in our case it holds because $J$ is dissipative and hence $\norm*{e^{Js}} \leq 1$ for all $s \geq 0$.} $\norm*{e^{Js} - e^{\tilde{J}s}} \leq \norm{J - \tilde J}s \leq \epsilon' s\norm{\mathcal{L}}_{\mathrm{be}}$. Further, \cref{lem:taylor} implies that the error caused by the truncation is $\frac{(2e^{s\norm{\mathcal{L}}_{\mathrm{be}}})^{K'+1}}{(K'+1)'}$. By assuming $\epsilon' \leq \epsilon/(t(\norm{\mathcal{L}}_{\mathrm{be}})$, the truncation error will dominate by our choice of $K'$.

  In \cref{lemma:block-encoding-channel}, we need $\ketbra{j}{j}\otimes A_j$, where $A_j$ is defined in \cref{eq:aj}. This can be implemented by a sequence of at most $K$ controlled-$U_{s}$ and at most $K$ controlled-$U_{L_j}$. Note that the time register required for implementing $U_{s}$ can be extracted from the index $j$, and then uncomputed. Therefore, the gate cost for this is $O(KK'm)$. Therefore, the additional 1- and 2-qubits for this implementation is dominated by $O(KK'm)$.

  Next, we analyze how the truncation of $e^{Js}$ at order $K'$ affects the total error. By \cref{lemma:jk-fk}, we have
  \begin{align}
    \begin{aligned}
      \|\mathcal{J}_{K'}(t-s_m) &\mathcal{L}_\mathrm{J} \mathcal{J}_{K'}(s_m-s_{m-1}) \mathcal{L}_\mathrm{J} \cdots \mathcal{J}_{K'}(s_2-s_1) \mathcal{L}_\mathrm{J} \mathcal{J}_{K'}(s_1) 
                                                            -\mathcal{F}_{k}(s_k, \ldots, s_1)\|_{\diamond} \\
                                  &\leq \frac{8e^{\norm{\mathcal{L}}_{\mathrm{be}}t}\norm{\mathcal{L}}_{\mathrm{be}}^{K'+1}}{(K'+1)!}(2\norm{\mathcal{L}}_{\mathrm{be}})^k2^kt^{K'+1}. 
    \end{aligned}
  \end{align}

  Taking the weighted sum for quadrature points, the error is at most
  \begin{align}
    \begin{aligned}
      \frac{8e^{\norm{\mathcal{L}}_{\mathrm{be}}t}\norm{\mathcal{L}}_{\mathrm{be}}^{K'+1}}{(K'+1)!}(2\norm{\mathcal{L}}_{\mathrm{be}})^k2^kt^{K'+1}\sum_{j_1=1}^q\cdots\sum_{j_k=1}^q\hat{w}_{(j_k)}\cdots \hat{w}_{(j_k,\ldots,j_1)} \\
      =  \frac{8t^ke^{\norm{\mathcal{L}}_{\mathrm{be}}t}\norm{\mathcal{L}}_{\mathrm{be}}^{K'+1}}{(k-1)!(K'+1)!}(2\norm{\mathcal{L}}_{\mathrm{be}})^k2^kt^{K'+1}.
    \end{aligned}
  \end{align}
  Therefore, the total error is
  \begin{align}
    \sum_{k=1}^K\frac{8t^ke^{\norm{\mathcal{L}}_{\mathrm{be}}t}\norm{\mathcal{L}}_{\mathrm{be}}^{K'+1}}{(k-1)!(K'+1)!}(2\norm{\mathcal{L}}_{\mathrm{be}})^k2^kt^{K'+1} \leq  \frac{32e^{5\norm{\mathcal{L}}_{\mathrm{be}}t}\norm{\mathcal{L}}_{\mathrm{be}}^{K'+2}t^{K'+2}}{(K'+1)!}.
  \end{align}
  With $t\norm{\mathcal{L}}_{\mathrm{be}} = \Theta(1)$, it suffices to set
  \begin{align}
    K' = O\left(\frac{\log(1/\epsilon)}{\log\log(1/\epsilon)}\right)
  \end{align}
  to make this error $\leq \epsilon/2$. Therefore, the total error is bounded by $\epsilon$.
  
  \textbf{Multiple simulation blocks.}
  For arbitrary evolution time $t$, we divide it into $O(t\norm{\mathcal{L}}_{\mathrm{be}})$ segments and set
  \begin{align}
    K, K', q = O\left(\frac{\log(t\norm{\mathcal{L}}_{\mathrm{be}}/\epsilon)}{\log\log(t\norm{\mathcal{L}}_{\mathrm{be}}/\epsilon)}\right)
  \end{align}
  so that the total error of the $O(t\norm{\mathcal{L}}_{\mathrm{be}})$ segments is within $\epsilon$. For the remaining smaller segment of this division, the normalizing constant is smaller which yields a larger success probability. However, the amplitude amplification will overshoot. We use standard technique by adding an ancillary qubit and use a rotation to dilute the success probability to 1/4. 
\end{proof}

\section{Conclusion and open questions}
In this paper, we presented a quantum algorithm for simulating Lindblad evolution, which captures the dynamics of Markovian open quantum systems. The algorithm can be used to forecast the dynamics of a quantum system interacting with an environment. Informally, the complexity of our algorithm scales as $\CO(t\,\mathrm{polylog}(t/\epsilon))$, which matches the previous state-of-the-art algorithm. Our algorithm is based on a conceptually novel mathematical treatment of evolution channel to preserve its complete positivity: we use a higher-order series expansion based on Duhamel's principle, and we approximate the integrals by scaled Gaussian quadrature, which exponentially reduces the number of terms in the summation. Our mathematical treatment trades off mathematical simplicity for technical conciseness, and it yields a much simpler algorithm based on linear combination of unitaries. We also outlined how our algorithm can be generalized to simulate time-dependent Lindbladians. Moreover, our approximation of multiple integrals using scaled Gaussian quadrature can be potentially used to produce a more efficient approximation of time-ordered integrals, which will simplify existing quantum algorithms for simulating time-dependent Hamiltonians based on a truncated Dyson series, e.g.,~\cite{KRS2011}.

The open questions of this work are summarized as follows.
\begin{itemize}
  \item Can we achieve the additive complexity, i.e., $\CO(t+\mathrm{polylog}(1/\epsilon))$? This additive complexity has been achieved for simulating Hamiltonian evolution by quantum signal processing~\cite{LC17} and quantum singular transformation~\cite{GSLW19}, and it is proved to be optimal~\cite{BCK15}. As Hamiltonian evolution is a special case of Lindblad evolution, the complexity for simulating the latter is at least $\Omega(t+\mathrm{polylog}(1/\epsilon))$. It is yet unknown how to generalize the techniques of quantum signal processing and quantum singular value transformation to superoperators.
  \item What are the practical performances of our algorithm? For Hamiltonian simulation, although LCU-based algorithms have a better asymptotic scaling, it was reported in~\cite{CMNRS18} that Trotter-based algorithms surprisingly perform just as well in practice. Regarding simulating Lindblad evolution, do LCU-based algorithms, i.e., the algorithms presented in this paper and~\cite{CW17}, have a practical advantage compared with Trotter-based simulation algorithms, e.g.,~\cite{KBG11,CL17}? An empirical study on the performances of quantum algorithms for simulating open quantum systems would be beneficial.
\end{itemize}

\section*{Acknowledgement}
We thank the anonymous reviewers for the valuable comments. We also thank Zecheng Li for pointing out various typos. XL's research is supported by the National Science Foundation Grants DMS-2111221. CW acknowledges support from National Science Foundation grant CCF-2238766 (CAREER). Both XL and CW were supported by a seed grant from the Institute of Computational and Data Science (ICDS).

\bibliographystyle{plain}
\bibliography{ref}

\appendix

\section{Simulating time-dependent Lindbladians}
\label{sec:timedependent}
There exist natural generalizations to Lindblad equations. One such generalization is time-dependent Markovian open quantum systems, which arises in the context of quantum heat engine~\cite{Alicki79,KL14,AG15,RK06} and controlling open quantum systems~\cite{Koch16,LKMKT18,SBMV18}. In this section, we sketch how our simulation techniques can be generalized to the case of time-dependent Lindbladians. More specifically, consider a time-dependent version of \cref{eq: lindb}:
\begin{align}
  \label{eq: lindb-td}
  \frac{\dd}{\dd t}\rho = \mathcal{L}(t)(\rho) \coloneqq -i[H(t), \rho] + \sum_{j=1}^m\left(L_j(t)\rho L_j^{\dag(t)} - \frac{1}{2}\left\{L_j(t)^{\dag}L_j(t), \rho\right\}\right).
\end{align}
Now, $H(t)$ and $L_j(t)$ are time-dependent. We decompose this time-dependent Lindbladian into drift terms and jump terms as:
\begin{align}\label{eq: LLJLL-td}
  \mathcal{L}(t) = \mathcal{L}_\mathrm{D}(t) &+ \mathcal{L}_\mathrm{J}(t), \quad \text{ and}\\
  \mathcal{L}_\mathrm{D}(t)(\rho) \coloneqq J(t)\rho + \rho J(t)^{\dag}, &\quad \mathcal{L}_\mathrm{J}(t)(\rho) = \sum_{j=1}^m L_j(t)\rho L_j(t)^{\dag}.
\end{align}
We express the evolution driven by $\mathcal{L}_\mathrm{D}$ as, 
\begin{align}
  \rho_t = \mathcal{V}(0, t) \coloneqq V(0, t) \rho_0  V(0,t)^\dagger,  
\end{align}
where $V(s,  t)$ satisfies the equation,
\begin{equation}
  \frac{d}{dt} V(s,t) = J(t) V(s,t), \quad \text{and } \quad V(s,s)=I.
\end{equation}
One can express the unitary $V(0, t)$ using time-ordered evolution operators,
\begin{align}
V(s,t)= \mathcal{T} e^{\int_s^t J(\tau) d\tau}.
\end{align}

Further, for \cref{eq: lindb}, the Duhamel's principle implies a generalization of \cref{eq: formula},
\begin{align}\label{eq: formula'}
  \rho_t =  \mathcal{V}(0,t)(\rho_0) + \int_0^t \mathcal{V}(s,t)(\mathcal{L}_\mathrm{J}(s) (\rho_s)) \dd s.
\end{align}

In the Hamiltonian simulation \cite{kieferova2019simulating},  Such an operator is approximated by Dyson series,
\begin{equation}
  V(0,t) = \sum_{k=0}^K \frac{t^k}{M^k k!} \sum_{j_1,j_2, \ldots, j_k=0}^{M-1} \mathcal{T} J(t_k) \cdots J(t_1) + \mathcal{O}\left( \frac{(\|J\|_{\max}  t)^{K+1} }{(K+1)!} + \frac{t^2 \norm*{\dot{J}}_{\max}  }{M}\right), 
\end{equation}
where $\mathcal{T}$ indicates a strict time-ordering $t_1\leq t_2\leq \cdots \leq t_k$ in the product.
The formula here approximates the evolution from $0$ to $t.$ This can be easily extended to another interval, due to the observation that,
\begin{align}
V(s,t)= \mathcal{T} e^{\int_s^t J(\tau) d\tau} = \mathcal{T} e^{\int_0^{t-s} J(s+\tau) d\tau},
\end{align}
which leads to
\begin{equation}
  V(s,t) = \sum_{k=0}^K \frac{t^k}{M^k k!} \sum_{j_1,j_2, \ldots, j_k=s}^{M-1}\mathcal{T} J(t_k) \cdots J(t_1) + \mathcal{O}\left( \frac{(\|J\|_{\max}  t)^{K+1} }{(K+1)!} + \frac{t^2 \norm*{\dot{J}}_{\max}  }{M}\right). 
\end{equation}

This suggests that, by repeatedly applying \cref{eq: formula}, we can adapt our series expansion in \cref{eq: GK} to
\begin{align}\label{eq: GKtd}
  \mathcal{G}_{K}(t) \coloneqq \mathcal{K}[V(0,t)] + \sum_{k=1}^K \int_{0 \leq s_1 \leq \cdots \leq s_k \leq t} \mathcal{F}_{k}(s_k, \ldots, s_1)\, \dd s_1 \cdots \dd s_k,
\end{align}
where
\begin{align}
  \label{eq:fktd}
  \begin{aligned}
   &\mathcal{F}_k(s_k, \ldots, s_1) \\
   &\coloneqq
   \mathcal{K}[V(s_k,t)] \mathcal{L}_\mathrm{J}(s_k) \mathcal{K}[V(s_{k-1}, s_k)] \mathcal{L}_\mathrm{J}(s_{k-1}) \cdots \mathcal{K}[V(s_1, s_2)] \mathcal{L}_\mathrm{J}(s_1) \mathcal{K}[V(0,s_1)].
  \end{aligned}
\end{align}

Then, we can approximate the integral using scaled Gaussian quadrature as in \cref{sec:quadrature}, and implement the completely positive map using the techniques presented in \cref{sec:alg}. Further note that we use a truncated Dyson series
\begin{equation}
    \mathcal{J}_K = \mathcal{K}\left[ \sum_{k=0}^K \frac{t^k}{M^k k!} \sum_{j_1,j_2, \cdots, j_k=s}^{M-1} J(t_k) \cdots J(t_1) \right].
\end{equation}
to approximate $V(s, t)$.

\section{Technical proofs related to higher-order series expansion}
\label{sec:proof-higherorder}

\higherorder*
\begin{proof}
  The error is the last term is \cref{eq: duhem}. Due to the fact that $\mathcal{L}_\mathrm{D}$ is dissipative, we have that $\norm{e^{\mathcal{L}_\mathrm{D}t}}_{\diamond} \leq 1.$ The last integral can be directly bounded by,
\begin{align}
  \begin{aligned}\label{eq: duham-err}
  &\left\|
  \int_{0 \leq s_1 \leq \cdots \leq s_{K+1} \leq t} e^{\mathcal{L}_\mathrm{D}(t-s_{K+1})} \mathcal{L}_\mathrm{J} e^{\mathcal{L}_\mathrm{D}(s_{K+1}-s_K)} \mathcal{L}_\mathrm{J} \cdots e^{\mathcal{L}_\mathrm{D}(s_2-s_1)} \mathcal{L}_\mathrm{J} (\rho_{s_1}) \dd s_1 \cdots \dd s_{K+1}\right\|_{\diamond} \\
  \leq & \int_{0 \leq s_1 \leq \cdots \leq s_{K+1} \leq t}  \dd s_1 \cdots \dd s_{K+1} 
  \norm{\mathcal{L}_\mathrm{J}}_{\diamond}^{K+1} = \frac{(\norm{\mathcal{L}_\mathrm{J}}_{\diamond} t)^{K+1}}{(K+1)!}.
  \end{aligned}
\end{align}
In addition, we have
\begin{align}
  \norm{\mathcal{L}_\mathrm{J}}_{\diamond} \leq \sum_{j=1}^m\norm{\mathcal{K}[L_j]}_{\diamond} = \sum_{j=1}^m\norm{\mathcal{K}[L_j\otimes I]}_1\leq \sum_{j=1}^m\norm{L_j}^2 \leq \sum_{j=1}^m\alpha_j^2 \leq 2\norm{\mathcal{L}}_{\mathrm{be}},
\end{align}
where we have used the inequality $\norm{ABC}_1 \leq \norm{A}\norm{B}_1\norm{C}$. The claimed bound is established.
\end{proof}

\jktaylor*
\begin{proof}
  Consider the remainder of the infinite series. We have
  By Taylor's theorem, 
  \begin{align}
    \label{eq:talor-rem}
    \norm{\sum_{\ell=k+1}^{\infty}\frac{J^{\ell}t^{\ell}}{\ell!}} \leq \frac{e^{\norm{J}t}\norm{J}^{k+1}t^{k+1}}{(k+1)!}.
  \end{align}

  It is also useful to obtain a tighter bound on the truncated series $\norm{\sum_{\ell=0}^k J^{\ell}t^{\ell}/\ell!}$. To achieve this, first observe that the eigenvalues of $J = -iH -\frac{1}{2}\sum_{j}L_j^{\dag}L_j$ have negative real parts. This implies that $\norm{e^{Jt}} \leq 1$. Using \cref{eq:talor-rem} and the triangle inequality, we have 
  \begin{align}
    \norm{\sum_{\ell=0}^{k}\frac{J^{\ell}t^{\ell}}{\ell!}} \leq \norm{e^{Jt}} + \norm{\sum_{\ell=k+1}^{\infty}\frac{J^{\ell}t^{\ell}}{\ell!}} \leq 2,
  \end{align}
  Note that we have used the assumption $(k+1)! \geq 2\norm{J}^{k+1}t^{k+1}$.

  Consider an arbitrary matrix $X \in \bbc^{2^n \times 2^n}$ with $\norm{X}_1 \leq 1$, we have
   \begin{align}
    \norm{\left(e^{\mathcal{L}_\mathrm{D}t} - \mathcal{J}_k(t)\right)(X)}_1 
                                                                   &\leq 2\norm{\sum_{\ell=0}^k\frac{J^{\ell}t^{\ell}}{\ell!}}\norm{\sum_{\ell=k+1}^{\infty}\frac{J^{\ell}t^{\ell}}{\ell!}} + \norm{\sum_{\ell=k+1}^{\infty}\frac{J^{\ell}t^{\ell}}{\ell!}}^2\\
                                                                   &\leq 4\norm{\sum_{\ell=k+1}^{\infty}\frac{J^{\ell}t^{\ell}}{\ell!}} + \norm{\sum_{\ell=k+1}^{\infty}\frac{J^{\ell}t^{\ell}}{\ell!}}^2\\
                                                                   &\leq \frac{4e^{\norm{J}t}\norm{J}^{k+1}t^{k+1}}{(k+1)!} + \left(\frac{e^{\norm{J}t}\norm{J}^{k+1}t^{k+1}}{(k+1)!}\right)^2 \\
                                                                   &\leq \frac{8e^{\norm{J}t}\norm{J}^{k+1}t^{k+1}}{(k+1)!},
  \end{align}
  where we have used the inequality $\norm{ABC}_1 \leq \norm{A}\norm{B}_1\norm{C}$.

  Now we extend this bound from the induced trace norm to the diamond norm by considering a larger Hilbert space. Suppose $\mathcal{L}_\mathrm{D}$ is acting on $\bbc^{2^n}$. Consider a Hilbert space $\bbc^{2^{n'}\times 2^{n'}}$ and let $\mathcal{I}_{2^{n'}}$ be the identity map acting on this space. We have
  \begin{align}
    (e^{\mathcal{L}_\mathrm{D}t} - \mathcal{J}_k(t))\otimes \mathcal{I}_{2^{n'}} = e^{(\mathcal{L}_\mathrm{D}\otimes \mathcal{I}_{2^{n'}})t} - \mathcal{K}\left[\sum_{\alpha=0}^k \frac{(J\otimes I)^{\alpha}t^{\alpha}}{\alpha!}\right].
  \end{align}
  Note that when $n' = n$, $\norm{\mathcal{L}_\mathrm{D}\otimes\mathcal{I}_{2^{n'}}}_1 = \norm{\mathcal{L}_\mathrm{D}}_{\diamond}$. As a result,
  \begin{align}
    \norm{e^{\mathcal{L}_\mathrm{D}t} - \mathcal{J}_k(t)}_{\diamond}  &= \norm{(e^{\mathcal{L}_\mathrm{D}t} - \mathcal{J}_k(t))\otimes \mathcal{I}_{2^{n'}}}_1 \\
                                                             &= \norm{e^{(\mathcal{L}_\mathrm{D}\otimes \mathcal{I}_{2^{n'}})t} - \mathcal{K}\left[\sum_{\alpha=0}^k \frac{(J\otimes I)^{\alpha}t^{\alpha}}{\alpha!}\right]}_1 \\
                                                             &\leq \frac{8e^{\norm{J}t}\norm{J\otimes I}^{k+1}t^{k+1}}{(k+1)!} \\
                                                             &= \frac{8e^{\norm{J}t}\norm{J}^{k+1}t^{k+1}}{(k+1)!}.
  \end{align}
  Then the claimed bound holds by the observation that
  \begin{align}
    \label{eq:norm-j-norm-be}
    \norm{J} \leq \norm{H} + \frac{1}{2}\sum_{j}\norm{L_j}^2 \leq \alpha_0 + \frac{1}{2}\sum_j\alpha_j^2 = \norm{\mathcal{L}}_{\mathrm{be}}.
  \end{align}
\end{proof}

\diamondnorm*
\begin{proof}
  By \cref{lem:taylor}, we know that
  \begin{align}
    \norm{\mathcal{K}[e^{Jt}]-\mathcal{J}_{K'}(t)}_{\diamond} \leq \frac{8e^{\norm{\mathcal{L}}_{\mathrm{be}}t}\norm{\mathcal{L}}_{\mathrm{be}}^{K'+1}t^{K'+1}}{(K'+1)!}.
  \end{align}
  By the triangle inequality, we have
  \begin{align}
    \label{eq:norm-jmt}
    \norm{\mathcal{J}_{K'}(t)}_{\diamond} \leq 1+ \frac{8e^{\norm{\mathcal{L}}_{\mathrm{be}}t}\norm{\mathcal{L}}_{\mathrm{be}}^{K'+1}t^{K'+1}}{(K'+1)!} \leq 2,
  \end{align}
  when we assume a reasonably large $K'$, i.e., $(K'+1)! \geq 8e^{\norm{\mathcal{L}}_{\mathrm{be}}t}\norm{\mathcal{L}}_{\mathrm{be}}^{K'+1}t^{K'+1}$.

  Consider each $1 \leq k \leq K$. Using a hybrid argument and triangle inequalities, we have
  \begin{align}
    \|&\mathcal{J}_{K'}(t-s_m) \mathcal{L}_\mathrm{J} \mathcal{J}_{K'}(s_m-s_{m-1}) \mathcal{L}_\mathrm{J} \cdots \mathcal{J}_{K'}(s_2-s_1) \mathcal{L}_\mathrm{J} \mathcal{J}_{K'}(s_1) 
                                                            -\mathcal{F}_{k}(s_k, \ldots, s_1)\|_{\diamond} \\
                                                               &\leq \sum_{j=0}^k\norm{\mathcal{J}_{M}(t- s_m)\mathcal{L}_\mathrm{J}\cdots \left(\mathcal{J}_{K'}(s_{j+1}-s_j)-\mathcal{K}[e^{J(s_{j+1}-s_j)}]\right)\mathcal{L}_\mathrm{J} \cdots \mathcal{K}[e^{Js_1}]}_{\diamond} \\
                                                                &\leq \sum_{j=0}^k\norm{\mathcal{L}_\mathrm{J}}_{\diamond}^k\cdot 2^j \cdot \frac{8e^{\norm{\mathcal{L}}_{\mathrm{be}}t}\norm{\mathcal{L}}_{\mathrm{be}}^{K'+1}(s_{j+1}-s_j)^{K'+1}}{(K'+1)!} \\
                                &\leq \frac{8e^{\norm{\mathcal{L}}_{\mathrm{be}}t}\norm{\mathcal{L}}_{\mathrm{be}}^{K'+1}}{(K'+1)!}\norm{\mathcal{L}_\mathrm{J}}_{\diamond}^k\sum_{j=0}^m 2^j (s_{j+1}-s_j)^{M+1} \\
                                &\leq \frac{8e^{\norm{\mathcal{L}}_{\mathrm{be}}t}\norm{\mathcal{L}}_{\mathrm{be}}^{K'+1}}{(K'+1)!}\norm{\mathcal{L}_\mathrm{J}}_{\diamond}^k2^k\left(\sum_{j=0}^m  s_{j+1}-s_j\right)^{K'+1} \\
                                &\leq \frac{8e^{\norm{\mathcal{L}}_{\mathrm{be}}t}\norm{\mathcal{L}}_{\mathrm{be}}^{K'+1}}{(K'+1)!}\norm{\mathcal{L}_\mathrm{J}}_{\diamond}^k2^kt^{K'+1} \\
                                &\leq \frac{8e^{\norm{\mathcal{L}}_{\mathrm{be}}t}\norm{\mathcal{L}}_{\mathrm{be}}^{K'+1}}{(K'+1)!}(2\norm{\mathcal{L}}_{\mathrm{be}})^k2^kt^{K'+1}. 
  \end{align}
\end{proof}

\section{Technical proofs related to scaled Gaussian quadrature}
\label{sec:proof-quadrature}

\derivativebound*
\begin{proof}
  We have
  \begin{align}
    \frac{\dd}{\dd s_{j}} \mathcal{F}_{k} &= \mathcal{K}[e^{J(t-s_k)}] \mathcal{L}_\mathrm{J} \cdots \frac{\dd}{\dd s_{j}}(\mathcal{K}[e^{J(s_{j+1}-s_{j})}] \mathcal{L}_\mathrm{J} \mathcal{K}e^{J(s_j-s_{j-1})}]) \mathcal{L}_\mathrm{J} \cdots \mathcal{K}[e^{J(s_1-0)}] \\
                                             & = \mathcal{K}[e^{J(t-s_k)}] \mathcal{L}_\mathrm{J} \cdots \frac{\dd}{\dd s_{j}}(\mathcal{K}[e^{J(s_{j+1}-s_{j})}]) \mathcal{L}_\mathrm{J} \mathcal{K}[e^{J(s_j-s_{j-1})}] \mathcal{L}_\mathrm{J} \cdots \mathcal{K}[e^{J(s_1-0)}]  \\
                                        &\quad + \mathcal{K}[e^{J(t-s_k)}] \mathcal{L}_\mathrm{J} \cdots \mathcal{K}[e^{J(s_{j+1}-s_{j})}] \mathcal{L}_\mathrm{J} \frac{\dd}{\dd s_j}(\mathcal{K}[e^{J(s_j-s_{j-1})}]) \mathcal{L}_\mathrm{J} \cdots \mathcal{K}[e^{J(s_1-0)}].
  \end{align}
  By induction, it is not hard to show that
  \begin{align}
    \label{eq:dfmn}
    \frac{\dd^k}{\dd s_j^{k'}}\mathcal{F}_{k} = \mathcal{K}[e^{J(t-s_k)}] \mathcal{L}_\mathrm{J} \cdots \left(\sum_{\ell=0}^{k'} \binom{k'}{\ell}\frac{\dd^{k'-\ell}}{\dd s_j^{k'-\ell}}(\mathcal{K}[e^{J(s_{j+1}-s_{j})}]) \mathcal{L}_\mathrm{J} \frac{\dd^\ell}{\dd s_j^\ell}(\mathcal{K}[e^{J(s_j-s_{j-1})}])\right) \mathcal{L}_\mathrm{J} \cdots \mathcal{K}[e^{J(s_1-0)}].
  \end{align}
  
  Now, we consider the $p$-th order derivatives of $\mathcal{K}[e^{J(s_{j+1}-s_j)}]$ and $\mathcal{K}[e^{J(s_{j}-s_{j-1})}]$. For any matrix $A$ with $\norm{A}_1 = 1$, we have
  \begin{align}
    \begin{aligned}
      \frac{\dd^p}{\dd s_j^p} &\mathcal{K}[e^{J(s_{j+1}-s_j)}](A) = 
      \sum_{\ell=0}^p \binom{p}{\ell}\left(\frac{\dd^{p-\ell}}{\dd s_j^{p-\ell}}e^{J(s_{j+1}-s_j)}\right) A 
      \left(\frac{\dd^{\ell}}{\dd s_j^{\ell}}e^{J(s_{j+1}-s_j)}\right)
    \end{aligned}
  \end{align}
  
  For the Kraus operators, we have the bound
  \begin{align}
    \label{eq:dkraus}
    \norm{\frac{\dd^{\ell}}{\dd s_j^{\ell}}e^{J(s_{j+1}-s_j)}} \leq \norm{J}^{\ell}.
  \end{align}
This gives the bound,
  \begin{align}
    \norm{\frac{\dd^p}{\dd s_j^p} \mathcal{K}[e^{(s_{j+1},s_j)}](A)}_1 \leq 2^{p}\norm{J}^p.
  \end{align}
  The above bound can be easily extended to the diamond norm by extending the Kraus operator in \cref{eq:dkraus} to a larger space by tensoring an identity operator. As a result, we have
  \begin{align}
    \label{eq:norm-djm}
    \norm{\frac{\dd^p}{\dd s_j^p} \mathcal{K}[e^{(s_{j+1},s_j)}]}_{\diamond} \leq 2^{p}\norm{J}^p.
  \end{align}
  
  Combining \cref{eq:norm-jmt,eq:dfmn,eq:norm-djm}, we have the stated error bound. 
  \begin{align}
    \norm{\frac{\dd^{k'}}{\dd s_j^{k'}}\mathcal{F}_{k}}_{\diamond} \leq 2^{k'}\norm{\mathcal{L}_\mathrm{J}}_{\diamond}^k(2\norm{J})^{k'} \leq 2^{2k'+k}\norm{\mathcal{L}}_{\mathrm{be}}^k\norm{J}^{k'}.
  \end{align}
\end{proof}

\scaledweight*
\begin{proof}
  Note that the Gaussian quadrature produces exact results when the integrand is a polynomial of degree up to $2q-1$. This implies that in \cref{eq: GQ}, $E_q[s^\ell]=0, ~ \forall  ~0 \leq \ell \leq 2q-1.$ Therefore, a direct  substitution yields \cref{eq:wsl}. \cref{eq:sum-w} is a special of \cref{eq:wsl}. In addition,
  \begin{align}
    \sum_{j_k=1}^q\cdots\sum_{j_{k-\ell}=1}^q\hat{w}_{(j_k)}\cdots\hat{w}_{(j_k,\ldots,j_{k-\ell})} &= \sum_{j_k=1}^q\cdots\sum_{j_{k-\ell}=1}^qw_{j_k} \cdot \frac{w_{j_{k-1}}\hat{s}_{j_k}}{t}\cdot \cdots \cdot \frac{w_{j_{k-\ell}}\hat{s}_{j_{k-\ell+1}}\cdots\hat{s}_{j_{k-1}}\hat{s}_{j_k}}{t^k} \\
                                                                                                    &=\frac{1}{t^k t^{k-1}\cdots t}\sum_{j_k=1}^q\cdots\sum_{j_{k-\ell}=1}^q w_{j_k}\hat{s}_{j_k}^{\ell} w_{j_{k-1}}\hat{s}_{j_{k-1}}^{\ell-1}\cdots w_{j_{k-\ell}} \\
                                                                                                    &=\frac{1}{t^k t^{k-1}\cdots t}\frac{t^{\ell+1}}{\ell+1}\frac{t^\ell}{\ell} \cdots \frac{t}{1}\\
                                                                                                    &= \frac{t^{\ell+1}}{(\ell+1)!},
  \end{align}
  where we have repeatedly used \cref{eq:wsl}. The above equation proves \cref{eq:ww}, which immediately implies \cref{eq:wwkto1}.
\end{proof}

\quadratureerror*
\begin{proof}
  We first discretize the last direction using the canonical quadrature points and weights as
  \begin{align}
    &\begin{aligned}
    &\left\|\int_{0 \leq s_1 \leq \cdots \leq s_k \leq t}\mathcal{F}_{k}(s_k, \ldots, s_1) \,\dd s_1\cdots \dd s_k  \right.\\
    &\left. -\quad \sum_{j_k=1}^{q}\int_{0 \leq s_1 \leq \cdots \leq s_{k-1} \leq \hat{x}_{(j_k)}} \mathcal{F}_{k}(\hat{x}_{(j_k)}, s_{k-1} \ldots, s_1)  \hat{w}_{(j_k)}\,\dd s_1\cdots \dd s_{k-1}\right\|_{\diamond} 
    \end{aligned} \\
    &\quad = O\left(\frac{t^{k-1}}{(k-1)!}\cdot\frac{2^{k+1}\norm{\mathcal{L}}_{\mathrm{be}}^k\norm{J}^{(2q)}t^{2q+1}q}{(2q)!}\right),
  \end{align}
  where  the error is in terms of the diamond norm. In the last step, we have used the fact that $\forall s>0$,
  \begin{align}
   \int_{0 \leq s_1 \leq \cdots \leq s_{k-1} \leq s } \dd s_1\cdots \dd s_{k-1} = \frac{s^{k-1}}{(k-1)!}.
  \end{align}
  
  For the next direction, we  treat each $\hat{x}_{(j_k)}$ as an upper bound and use a scaled quadrature points and weights as follows.
  \begin{align}
    &\begin{aligned}
        &\left\|\sum_{j_k=1}^{q}\int_{0 \leq s_1 \leq \cdots \leq s_{k-1} \leq \hat{x}_{(j_k)}} \mathcal{F}_{k}(\hat{x}_{(j_k)}, s_{k-1} \ldots, s_1)  \hat{w}_{(j_k)}\,\dd s_1\cdots \dd s_{k-1} - \right.\\
     &\left.\quad \sum_{j_{k-1}=1}^q\sum_{j_k=1}^{q}\int_{0 \leq s_1 \leq \cdots \leq s_{k-2} \leq \hat{x}_{(j_k, j_{k-1})}}\mathcal{F}_{k}(\hat{x}_{(j_k)}, \hat{x}_{(j_k,j_{k-1})}, s_{k-2}, \ldots, s_1)  \hat{w}_{(j_k)} \hat{w}_{(j_k,j_{k-1})}\,\dd s_1\cdots \dd s_{k-2}\right\|_{\diamond} 
  \end{aligned} \\
    &\quad = O\left(\left(\frac{t^{k-2}}{(k-2)!}\sum_{j_k=1}^q\hat{w}_{(j_k)}\right)\cdot\frac{2^{k+1}\norm{\mathcal{L}}_{\mathrm{be}}^k\norm{J}^{(2q)}t^{2q+1}q}{(2q)!}\right) \\
    &\quad = O\left(\frac{t^{k-1}}{(k-2)!}\cdot\frac{2^{k+1}\norm{\mathcal{L}}_{\mathrm{be}}^k\norm{J}^{(2q)}t^{2q+1}q}{(2q)!}\right), 
  \end{align}
  where the last equation follows from \cref{eq:sum-w}.  
  Keep applying the Gaussian quadrature rule for each direction, we have for all $1 \leq \ell \leq k-1$,
  \begin{scriptsize}
  \begin{align}
    &\begin{aligned}
    &\left\|\sum_{j_{k-\ell}=1}^q\cdots\sum_{j_k=1}^q\int_{0 \leq s_1 \leq \cdots \leq s_{k-\ell-1} \leq \hat{x}_{(j_k,\cdots,j_{k-\ell})}}\mathcal{F}_{k}(\hat{x}_{(j_k)}, \ldots, \hat{x}_{(j_k,\ldots,j_{k-\ell})}, s_{k-\ell-1}, \ldots, s_1)\right.\\
    &\quad\quad\quad\quad\quad\quad\quad \left. \hat{w}_{(j_k)}\cdots\hat{w}_{(j_k,\ldots,j_{k-\ell})}\,\dd s_1\cdots \dd s_{k-\ell-1}  \right. \\
    &\left.\quad - \sum_{j_{k-\ell-1}=1}^q\cdots\sum_{j_k=1}^{q}\int_{0 \leq s_1 \leq \cdots \leq s_{k-\ell-2} \leq \hat{x}_{(j_k,\ldots,j_{k-\ell-1})}} \mathcal{F}_{k}(\hat{x}_{(j_k)}, \ldots, \hat{x}_{(j_k,\ldots,j_{k-\ell-1})}, s_{k-\ell-2}, \ldots, s_1)\right. \\
    &\quad\quad\quad\quad\quad\quad\quad \left.\hat{w}_{(j_k)}\cdots\hat{w}_{(j_k,\ldots,j_{k-\ell-1})}\,\dd s_1\cdots \dd s_{k-\ell-2} \right\|_{\diamond}
    \end{aligned} \\
    &\quad = O\left(\left(\frac{t^{k-\ell-2}}{(k-\ell-2)!}\sum_{j_{k}=1}^q\hat{w}_{(j_{k})}\cdots\sum_{j_{k-\ell}=1}^q\hat{w}_{(j_k\ldots,j_{k-\ell})}\right)\cdot\frac{2^{k+1}\norm{\mathcal{L}}_{\mathrm{be}}^k\norm{J}^{(2q)}t^{2q+1}q}{(2q)!}\right) \\
    &\quad = O\left(\left(\frac{t^{k-\ell-2}}{(k-\ell-2)!}\cdot\frac{t^{\ell+1}}{(\ell+1)!}\right)\cdot\frac{2^{k+1}\norm{\mathcal{L}}_{\mathrm{be}}^k\norm{J}^{(2q)}t^{2q+1}q}{(2q)!}\right) \\
    &\quad = O\left(\left(\frac{t^{k-1}}{(k-\ell-2)!(\ell+1)!}\right)\cdot\frac{2^{k+1}\norm{\mathcal{L}}_{\mathrm{be}}^k\norm{J}^{(2q)}t^{2q+1}q}{(2q)!}\right),
  \end{align}
  \end{scriptsize}
  where the second equality follows from \cref{eq:ww}.  

  Putting everything together and using the triangle inequality, we have
    \begin{align}
    &\begin{aligned}
    &\left\|\int_{0 \leq s_1 \leq \cdots \leq s_k \leq t}\mathcal{F}_{k}(s_k, \ldots, s_1) \,\dd s_1\cdots \dd s_k -\right. \\
    &\left.\quad-\sum_{j_1=1}^{q} \cdots \sum_{j_k=1}^{q}\mathcal{F}_{k}\left(\hat{x}_{(j_k)}, , \ldots, \hat{x}_{(j_k, \ldots, j_1)}\right)\hat{w}_{(j_k)}\cdots \hat{w}_{(j_k, \ldots, j_1)}\right\|_{\diamond} 
    \end{aligned} \\
    &\quad = O\left(\left(\frac{t^{k-1}}{(k-1)!}+\sum_{\ell=0}^{k-2}\frac{t^{k-1}}{(k-\ell-2)!(\ell+1)!}\right)\frac{2^{k+1}\norm{\mathcal{L}}_{\mathrm{be}}^k\norm{J}^{(2q)}t^{2q+1}q}{(2q)!}\right) \\
    &\quad = O\left(\left(\sum_{\ell=-1}^{k-2}\frac{t^{k-1}}{(k-\ell-2)!(\ell+1)!}\right)\frac{2^{k+1}\norm{\mathcal{L}}_{\mathrm{be}}^k\norm{J}^{(2q)}t^{2q+1}q}{(2q)!}\right) \\
    &\quad = O\left(\frac{(2t)^{k-1}2^{k+1}\norm{\mathcal{L}}_{\mathrm{be}}^k\norm{J}^{(2q)}t^{2q+1}q}{(k-1)!(2q)!}\right).
  \end{align}
\end{proof}

\end{document}